\documentclass[letterpaper,11pt]{article}
\usepackage[margin=1in]{geometry}

\usepackage{mathptmx} 
\DeclareMathAlphabet\mathcal{OMS}{cmsy}{m}{n}
\SetMathAlphabet\mathcal{bold}{OMS}{cmsy}{b}{n}

\usepackage[utf8]{inputenc}
\usepackage{amsthm,amsmath,amssymb,mathtools}
\usepackage{paralist}
\usepackage{xspace}
\usepackage{authblk}
\usepackage{enumitem}

\usepackage{tikz}
\usepackage{tikz-3dplot}
\usetikzlibrary{shapes,calc,positioning,intersections, angles,quotes}
\tdplotsetmaincoords{80}{120}

\usepackage[numbers]{natbib}
\usepackage{graphicx}
\usepackage{subcaption}
\usepackage{csquotes}
\usepackage{hyperref}
\usepackage{doi}

\newtheorem{theorem}{Theorem}
\newtheorem{lemma}[theorem]{Lemma}
\newtheorem{corollary}[theorem]{Corollary}
\newtheorem{Definition}[theorem]{Definition}
\newtheorem{observation}[theorem]{Observation}
\theoremstyle{remark}

\newcommand{\eps}{\varepsilon}

\newcommand{\opt}{\textsc{Opt}\xspace}
\newcommand{\R}{\mathbb{R}}
\newcommand{\normvect}{\vec{n}}

\newcommand{\separate}{\vec{s}}
\newcommand{\dotprod}[2]{\left\langle #1,#2\right\rangle}

\newcommand{\calI}{\mathcal{I}}

\newcommand{\shift}{\rho}

\newcommand{\grid}{\Gamma}
\newcommand{\gridpoints}{G}
\newcommand{\gridpoint}{p}
\newcommand{\scalevar}{\lambda}
\newcommand{\translvar}{\shift}

\newcommand{\val}{c}

\DeclareMathOperator{\para}{polytopes}
\DeclareMathOperator{\conv}{conv}
\DeclareMathOperator{\tsp}{TSP}
\DeclareMathOperator{\vertices}{vertices}
\DeclareMathOperator*{\argmax}{argmax}
\DeclareMathOperator*{\argmin}{argmin}

\usepackage{xcolor}\definecolor{darkgreen}{rgb}{0,0.392,0}
\usepackage{todonotes}
\newcounter{note}[section] \renewcommand{\thenote}{\thesection.\arabic{note}}
\def\DEBUG{true}

\ifdefined\DEBUG

 \newcommand{\kevin}[1]{\refstepcounter{note}\textcolor{blue}{%
     \mathversion{bold}\marginpar{\hfill\tiny\sffamily\bfseries
       \textcolor{blue}{KS~\thenote}}$\ll$\bfseries\sffamily#1
     --Kevin$\gg$}\mathversion{normal}}
 \newcommand{\antonios}[1]{\refstepcounter{note}\textcolor{red}{%
     \mathversion{bold}\marginpar{\hfill\tiny\sffamily\bfseries
       \textcolor{red}{AA~\thenote}}$\ll$\bfseries\sffamily#1
     --Antonios$\gg$}\mathversion{normal}}
  \newcommand{\krzysztof}[1]{\refstepcounter{note}\textcolor{darkgreen}{%
     \mathversion{bold}\marginpar{\hfill\tiny\sffamily\bfseries
       \textcolor{darkgreen}{KF~\thenote}}$\ll$\bfseries\sffamily#1
     --Krzysztof$\gg$}\mathversion{normal}}
 \newcommand{\ruben}[1]{\refstepcounter{note}\textcolor{orange}{%
     \mathversion{bold}\marginpar{\hfill\tiny\sffamily\bfseries
       \textcolor{orange}{RH~\thenote}}$\ll$\bfseries\sffamily#1
     --Ruben$\gg$}\mathversion{normal}}

\else
\newcommand{\kevin}[1]{}
\newcommand{\antonios}[1]{}
\newcommand{\krzysztof}[1]{}
\newcommand{\ruben}[1]{}

\usepackage[disable]{todo}

\fi

\title{A PTAS for Euclidean TSP with Hyperplane Neighborhoods\footnote{A preliminary version of this paper appeared in the Proceedings of the Thirtieth Annual ACM-SIAM Symposium on Discrete Algorithms (SODA 2019).}}
\author[1]{Antonios Antoniadis\thanks{Supported by
    Deutsche Forschungsgemeinschaft (DFG) grant AN 1262/1-1.}} 
\author[2]{Krzysztof Fleszar\thanks{Supported by CONICYT Grant PII 20150140 and ERC consolidator grant TUgbOAT no.\ 772346.}} 
\author[3]{Ruben Hoeksma} 
\author[4]{Kevin Schewior\thanks{Supported by CONICYT Grant PII 20150140 and DAAD PRIME program.}}
\affil[1]{Saarland University and Max-Planck-Institut f\"ur
  Informatik, Saarland University Campus, Saarbr\"ucken, Germany.
  \texttt{aantonia@mpi-inf.mpg.de}}
\affil[2]{University of Warsaw, Warsaw, Poland. \texttt{kfleszar@mimuw.edu.pl}}
\affil[3]{Universit\"at Bremen, Bremen, Germany. \texttt{hoeksma@uni-bremen.de}}
\affil[4]{Technische Universit\"at M\"unchen, M\"unchen, Germany. \texttt{kschewior@gmail.com}}

\begin{document}
\maketitle
\begin{abstract}
  \setcounter{page}{0}
  In the Traveling Salesperson Problem with Neighborhoods (TSPN),
  we are given a collection of geometric regions in some space. The goal
  is to output a tour of minimum length that visits at least one point
  in each region. Even in the Euclidean plane, TSPN is known to be APX-hard
~\cite{TSPNpointpairAPX}, 
  which gives rise to studying more tractable special cases of the problem.
In this paper, we focus on the fundamental special case of regions that are hyperplanes in the~$d$-dimensional Euclidean space. This case contrasts the much-better understood case of so-called fat regions~\cite{TSPNdoublingMetrics, sodaFatRegions} .
    
  While for~${d=2}$ an exact
  algorithm with a running time of~${O(n^5)}$ is
  known~\cite{TSPNlinesPlaneJonsson}, settling the exact approximability of the
  problem for~${d=3}$ has been repeatedly posed as an open
  question~\cite{TSPNplaneJoA, TSPNlinesBallsPlanes, 
sodaFatRegions, handbookgeom}. To date, only an
  approximation algorithm with guarantee exponential in~$d$ is
  known~\cite{TSPNlinesBallsPlanes}, and NP-hardness remains~open.
  
  For arbitrary fixed~$d$, we develop a Polynomial Time Approximation
  Scheme (PTAS) that works for both the tour and path version of the
  problem. Our algorithm is based on approximating the convex hull of
  the optimal tour by a convex polytope of bounded complexity. 
  After enumerating a number of structural properties
 of these polytopes, a linear program 
  finds one of them that minimizes the length of the tour.
  As the approximation guarantee approaches~$1$, 
  our scheme adjusts the complexity of the considered polytopes accordingly.
  
  In the analysis of our approximation scheme, we show that our search space includes a sufficiently
  good approximation of the optimum. To do so, we develop a novel and general
  sparsification technique that transforms an arbitrary convex polytope
  into one with a constant number of vertices, and, subsequently, into one
  of bounded complexity in the above sense. 
  We show that this transformation 
  does not increase the tour length by too much, while the transformed tour visits
  any hyperplane that it visited before the~transformation.
  
\end{abstract}

\thispagestyle{empty}
\pagebreak
\section{Introduction}

The Traveling Salesperson Problem (TSP) is commonly regarded as one of the most
important problems in combinatorial optimization. In TSP, a salesperson wishes
to find a tour that visits a set of clients in the shortest way possible. It is
very natural to consider metric TSP, that is, to assume that the clients are
located in a metric space. While Christofides' famous
algorithm~\cite{Christofides} then attains an approximation factor of~${3/2}$, the
problem is APX-hard even under this assumption~\cite{PapadimitriouY93}. This
lower bound and the paramount importance of the problem has motivated the study
of more specialized cases, in particular Euclidean TSP (ETSP), that is, metric
TSP where the metric space is Euclidean. ETSP admits a Polynomial Time
Approximation Scheme (PTAS), a~${(1+\eps)}$-approximation polynomial time
algorithm for any fixed~${\eps>0}$, which we know from the celebrated results of
Arora~\cite{AroraPTAS} and Mitchell~\cite{MitchellPTAS}. These results have
subsequently been improved and generalized~\cite{DBLP:conf/focs/BartalG13, 
DBLP:journals/siamcomp/BartalGK16, RaoSmithPTASimprovements}.

A very natural generalization of metric TSP is motivated by clients that are
not static (as in TSP) but willing to move in order to meet the salesperson. In the
Traveling Salesperson Problem with Neighborhoods (TSPN), first studied by Arkin and
Hassin in the Euclidean setting~\cite{ArkinHassinTSPN}, we are given a set of reasonably represented
(possibly disconnected) regions. The task is to compute a minimum-length tour
that visits these regions, that is, the tour has to contain at least one point
from every region.  In contrast to regular TSP, the problem is already APX-hard
in the Euclidean plane, even for neighborhoods of relatively low
complexity~\cite{TSPNpointpairAPX,TSPNsegmentAPX}. Whereas the problem did
receive considerable attention and a common focus was identifying natural
conditions on the input that admit a PTAS, the answers that were found are
arguably not yet satisfactory. For instance, it is not known whether the special
case of disjoint connected neighborhoods in the plane is
APX-hard~\cite{sodaFatRegions,handbookgeom}. On the other hand, there has been a
line of work~\cite{CorridorConnectionProblem, Chan2011, TSPNdoublingMetrics, 
TSPNplaneJoA, sodaFatRegions} that has led up to a PTAS for ``fat" regions in
the plane~\cite{sodaFatRegions} and a restricted version of such regions
(``weakly disjoint") in general doubling metrics~\cite{TSPNdoublingMetrics}.
Here, a region is called fat if the radii of the largest ball included in the
region and the smallest ball containing the region are within a constant factor
of each other.

In this paper, we focus on the fundamental case in which all regions are
hyperplanes (in Euclidean space of fixed dimension~$d$) and give a PTAS (more
precisely, EPTAS), improving upon a~${2^{\Theta(d)}}$-approximation~\cite{TSPNlinesBallsPlanes}. Not only is the
problem itself considered \enquote{particularly intriguing}~\cite{TSPNplaneJoA}
and has its complexity status been repeatedly posed as an open problem~\cite{TSPNplaneJoA, TSPNlinesBallsPlanes, sodaFatRegions,
	handbookgeom}.
 It also seems plausible that studying this problem, which is
somewhat complementary to the much-better understood case of fat regions, will
add techniques to the toolbox for TSPN that may lead towards understanding which
cases are tractable in an approximation sense. Indeed, our techniques are novel
and quite general: Using a sparsification technique, we show that a certain
class of bounded-complexity polytopes can be parameterized to represent the
optimal solution well enough. To compute a close approximation to that polytope,
we boost the computational power of an LP by enumerating certain crucial
properties of the polytope.

\paragraph{Further Related Work.} 
In contrast to regular TSP, TSPN is already
APX-hard in the Euclidean plane~\cite{TSPNvaryingSizeAPXhardness}. For some
cases in the Euclidean plane, there is even no polynomial-time~${O(1)}$-approximation (unless P~${=}$ NP), for instance, the case where each region
is an arbitrary finite set of points~\cite{TSPNcomplexity} (Group TSP).
The problem remains APX-hard when there are exactly two points in each
region~\cite{TSPNpointpairAPX} or the regions are line segments of similar
lengths~\cite{TSPNsegmentAPX}.

Positive results for TSPN in the Euclidean plane were obtained in the seminal
paper of Arkin and Hassin~\cite{ArkinHassinTSPN}, who gave~${O(1)}$-approximation
algorithms for various cases of bounded neighborhoods, including translates of
convex regions and parallel unit segments. The only known approximation
algorithm for~$n$ general bounded neighborhoods (polygons) in the plane is an~${O(\log n)}$-approximation~\cite{DBLP:conf/compgeom/MataM97}. Partly in more
general metrics,~${O(1)}$-approximation algorithms and approximation schemes were
obtained for other special cases of bounded regions, which are disjoint,
fat,
or of comparable sizes~\cite{TSPNvaryingSizeAPXhardness, CorridorConnectionProblem,
	TSPNdoublingMetrics, TSPNplaneJoA, sodaFatRegions,
	Mitchell2010-pairwisedisjoint-connected}.

We review results for the case of unbounded neighborhoods, such as lines or planes.
For~$n$ lines in the plane, the problem can be solved exactly in~${O(n^5)}$ time
by a reduction to the watchman route problem~\cite{TSPNlinesPlaneJonsson} and
using existing algorithms for the latter problem~\cite{watchman99,
watchman03,watchman08,watchmanCorrigendum}. 
A~${1.28}$-approximation is possible in linear time~\cite{TSPNlinesRays}. 
This result uses that the smallest rectangle enclosing the optimal tour is
already a good approximation. 
By a straightforward reduction from ETSP in the plane, the problem becomes
NP-hard if we consider lines in three dimensions. For the latter case, only a
recent~${O(\log^3 n)}$-approximation algorithm by Dumitrescu and
Tóth~\cite{TSPNlinesBallsPlanes} is known. They tackle the problem by reducing
it to a group Steiner tree instance on a geometric graph while losing a constant factor. Then they apply a known approximation
algorithm for group Steiner tree. If neighborhoods are planes in 3D, or
hyperplanes in higher constant dimensions, then it is even open whether the problem
is NP-hard. Only one known approximation
result has been obtained so far: The linear-time algorithm of Dumitrescu and
Tóth~\cite{TSPNlinesBallsPlanes} finds, for any constant dimension~$d$ and any
constant~${\eps>0}$, a~${(1+\eps)2^{d-1}/\sqrt{d}}$-approximation of the optimal
tour. Their algorithm generalizes the ideas used for the two-dimensional case~\cite{TSPNlinesRays}.
Via a low-dimensional LP, they find a~${(1+\eps)}$-approximation of the smallest
box enclosing the optimal tour. Then they output a Hamiltonian cycle on the
vertices of the box as a solution. They observe that any tour visiting all the
vertices of the box is a feasible solution and that the size of the box is
similar to the length of the optimal tour. This allows them to relate the length
of their solution to the length of the optimal tour. For the three-dimensional case and a sufficiently
small~$\eps$, their algorithm gives a~${2.31}$-approximation.

Observe that all of the above approximation results hold -- with a loss of a factor of~$2$ --
also for the TSP path problem where the goal is to find a shortest path visiting
all regions (with arbitrary start and end point). 
For the case of lines in the plane, there is a~${1.61}$-approximation linear-time
algorithm~\cite{TSPNlinesRays}. 

For improving the results on hyperplane neighborhoods, a repeatedly expressed
belief is the following: If we identify the smallest convex region~$C$
intersecting all hyperplanes, then we can scale it up by a polynomial factor to a region~${C'}$ such that~${C'}$ contains the optimal tour. Interestingly, Dumitrescu and
Tóth~\cite{TSPNlinesBallsPlanes} refute this belief by giving an example where
no~${(1+\eps)}$-approximate tour exists within such a region~${C'}$, for a small
enough constant~${\eps>0}$. This result makes it unlikely that first narrowing
down the search space to a bounded region (such as the box computed in the~${2^{\Theta(d)}}$-approximation by Dumitrescu and
T\'{o}th~\cite{TSPNlinesBallsPlanes}) and then applying local methods is a
viable approach to obtaining a PTAS. Indeed, the technique that we present in
this paper is much more global.

\paragraph{Our Contribution and Techniques.} 
The main result of this paper is a PTAS for TSP with hyperplane neighborhoods in
fixed dimensions. For fixed~$\eps$ and~$d$, our algorithm runs in strongly polynomial linear time 
(that is, the number of arithmetic operations is bounded linearly in the number of hyperplanes 
and the space is bounded by a polynomial in the input length), making it an EPTAS (\emph{efficient} PTAS). This is a significant step towards
settling the complexity status of the problem, which had been posed as an open
problem several times over the past 15 years~\cite{TSPNplaneJoA,
	TSPNlinesBallsPlanes, sodaFatRegions, handbookgeom}. %
\begin{theorem}\label{thm:main}
	For every fixed~${d\in\mathbb{N}}$ and~${\eps>0}$, there is
	a~${(1+\eps)}$-approximation algorithm for TSP with hyperplane neighborhoods
	in~${\mathbb{R}^d}$ that runs in strongly polynomial linear time.
\end{theorem}

Our technique is based on the observation that the optimal tour~${T^\star}$ can be
viewed as the shortest tour visiting all the vertices of a certain polytope~${P^\star}$, the convex hull of~${T^\star}$. So, in order to approximate the optimal
tour, one may also think about finding a convex polytope with a short and
feasible tour on its vertices that intersects all hyperplanes. In this light, the~${({(1+\eps)2^{d-1}/\sqrt{d}})}$-approximation by Dumitrescu and T\'oth~\cite{TSPNlinesBallsPlanes},
which, by using an LP, finds a cuboid with minimum perimeter intersecting all input
hyperplanes, can be viewed as a very crude approximation of~${P^\star}$. Note that
forcing the polytope to intersect all hyperplanes makes each tour on its
vertices feasible.

The approach we take here can be viewed as an extension of this idea. Namely, we also
use an LP with~${O_{\eps,d}(1)}$ (that is,~${O(1)}$ when~$\eps$ and~$d$ are constants) many variables
to find bounded-complexity polytopes intersecting all input
hyperplanes. However, the extension to a~${(1+\eps)}$-approximation raises three main challenges:
\begin{compactenum}
	\item In order to get a~${(1+\eps)}$-approximation, the complexity %
	of the polytope increases to an arbitrarily high level as~${\eps\rightarrow 0}$.
	We need to come up with a suitable definition of complexity. 
	\item More careful
	arguments are necessary for comparing the optimum with the shortest tour among
	the feasible ones that visit all vertices of a polytope of bounded complexity.
	\item As the complexity of the polytope increases, we need to handle more and
	more complicated combinatorics, which makes writing an LP more difficult.
\end{compactenum}
 
In this paper, we overcome all three challenges. On the way, we introduce
several novel ideas, many of which can be of independent interest, and combine
known ones.

First, we define polytopes of bounded complexity to be those that can be
obtained in the following way. From the integer grid~${\{0,\dots,g\}^d}$ (for some
suitably chosen~${g\in O_{\eps,d}(1)}$), select a subset of points and take the
convex hull of them. The result is called a \emph{base polytope}. Subsequently
translate and uniformly scale the base polytope arbitrarily to obtain the final
polytope.

The second challenge is overcome by turning~${P^\star}$, the convex hull of an
optimum tour~${T^\star}$, into one of the polytopes of bounded complexity without
increasing the length of the shortest tour on the vertices by more than a~${(1+\eps)}$-factor. The most straightforward way of getting such a polytope that is
``similar" to~${P^\star}$ is the following: Take the smallest axis-aligned
hypercube that includes~${P^\star}$ and subdivide it uniformly by an appropriate translated and scaled copy of our integer grid~ $\{0,\dots,g\}^d$.
Now, for each vertex~${v\in P^\star}$, take all
vertices of the grid cell containing~$v$, and take the convex hull of all these
points to obtain~$P$. Clearly,~$P$ is of bounded complexity as defined above.

However, in order to satisfactorily bound the length of the shortest tour on the
vertices of~$P$ with  
respect to~${|T^\star|}$, we need~${P^\star}$ to have
only 
few vertices. For instance, if~${P^\star}$ had~${k=O_{\eps,d}(1)}$ vertices, 
we could choose the granularity of the grid to be small enough so that we could
transform~${T^\star}$ into a tour of the vertices of~$P$ by making it longer by
only the additive length of~${\eps\cdot|T^\star|/k}$ at each vertex. Since in
general we cannot bound the number of vertices of~${P^\star}$, we first transform~${P^\star}$ into an intermediate polytope~${P'}$ that has~${O_{\eps,d}(1)}$ vertices,
and only then do we apply the above construction to obtain~$P$. This is where the
following structural result, which is likely to have more general applications,
is used.

For a general polytope, we show how core sets~\cite{coreset-survey,Chan06} can
be used to select~${O_{\eps,d}(1)}$ many of its vertices 
such that if we scale the convex hull of these selected vertices by the
factor~${1+\eps}$, from some carefully chosen center, the scaled convex hull
contains the original polytope. 
This result comes in handy, because we can scale~${T^\star}$ in the same way to
obtain a tour of the vertices of~${P'}$ of length~${(1+\eps)\cdot|T^\star|}$. The
proof utilizes properties of the maximum inscribed hyper-ellipsoid, due to
John~\cite{John1948} (see also the refinement due to Ball~\cite{Ball1992} that
we use in this paper). 

The third challenge is to find the tour constructed
above by using linear programing. The idea is the following:
We enumerate all base polytopes. For each base polytope, we write an LP that
finds the shortest feasible tour on the vertices of a polytope obtained from the
base polytope by uniform scaling and translating: The LP has~${d+1}$ variables:~$d$ variables for translating the base polytope and another one for scaling; an
assignment of these variables naturally corresponds to a polytope. The objective
function is the length of the shortest tour on the vertices of that polytope,
which can be computed by multiplying the value of the scaling variable with the
(pre-computed) length of the shortest tour on the vertices of the base polytope.
To force this tour to be feasible, we use an idea by Dumitrescu and
T\'oth~\cite{TSPNlinesBallsPlanes}: By convexity of the polytope, for each input
hyperplane, we can identify two vertices, the \emph{separated pair}, that are on
different sides of the hyperplane if and only if the polytope intersects the
input hyperplane. We write constraints that make sure that this is the case.

We note that our techniques easily extend to the path variant of the problem and several variants with prespecified parameters of the tour. We discuss these in Section~\ref{sec:discussion}.

\paragraph{Overview of this Paper.} %
In Section~\ref{sec:prelim}, we introduce some notation that we use throughout
the paper and make some preliminary observations. Then, in
Section~\ref{sec:struct}, we show that the shortest TSP tour that satisfies
certain conditions is a~${(1+\eps)}$-approximation of the overall shortest
TSP tour. In Section~\ref{sec:algo}, we describe an algorithm that computes a~${(1+\eps)}$-approximation of the shortest TSP tour that satisfies these
conditions.
Finally, in Section~\ref{sec:discussion}, we discuss remaining open
problems and the implication of our work for the TSPN path problem with
hyperplanes and some extensions.

\section{Preliminaries}
\label{sec:prelim}
\paragraph{Problem Definition.}
Throughout this paper, we fix a dimension~$d$ and restrict ourselves
to the Euclidean space~${\mathbb{R}^d}$. The input of TSPN for
hyperplanes consists of a set~$\mathcal{I}$ of~$n$
hyperplanes. 
Every hyperplane is given by~$d$ integers~${(a_i)_{i=1}^d}$,
where not all scalars are~$0$, and an integer~$c$, and it contains all
points~$x$ 
that satisfy~${a_1x_1+\dots + a_nx_n = c}$. A \emph{tour} is a closed
polyline 
and is called \emph{feasible} or \emph{a feasible solution} if it
visits every hyperplane of~$\mathcal{I}$, that is, if it contains a
point in every hyperplane of~$\mathcal{I}$.
A tour is \emph{optimal} or \emph{an optimal solution} if it is a
feasible tour of minimum length. 
 The goal is to find an optimal tour. 
 Given~$\mathcal{I}$, we call any such optimal tour~$\opt$ or~${\opt(\mathcal{I})}$. 
The length of a tour~$T$ is given by~${|T|}$.

\paragraph{Computational Model.} We use a slight extension of the real
RAM as the 
computational model. In addition to the standard arithmetic operations, we assume that the operation of computing integer square roots takes constant time.
This does not significantly increase our computational power as an integer square root can be computed on a Turing machine in essentially the same time as multiplying two numbers of the same magnitude~\cite{Brent1976}.

\paragraph{Notation.}
For any positive integer~$g$, we define the integer grid~$\grid_g$ as
the subset~${\{0,\dots,g\}^d}$ of~${\mathbb{N}^d}$ (that is, the subset
of~${\mathbb{N}^d}$ containing all points whose coordinate values are 
smaller or equal to some integer~$g$). The \emph{size} of the grid is~$g$.
We say that a convex polytope is a \emph{base polytope of\ \ $\grid_g$} if 
its vertices are grid points of~$\grid_g$. 
By~${\para(\grid_g{})}$, we denote the set of all polytopes that can be constructed 
by \emph{scaling and translating a base polytope of~$\grid_g$}.

For a polyhedron~$P$,~${\vertices(P)}$ denotes the set of its vertices. A
\emph{tour of a point set~$\mathcal{P}$} is a closed polyline that contains
every point of~$\mathcal{P}$. 
Throughout this paper, let~${\tsp(\mathcal{P})}$ denote any shortest tour
of~$\mathcal{P}$, and~${\conv(\mathcal{P})}$ denote the convex hull
of~$\mathcal{P}$. For a tour~$T$,~${\conv(T)}$ denotes the convex hull of the
tour. We also use~${\tsp(P)}$ for a polyhedron~$P$ in order to refer to
the \emph{tour of a polyhedron}, that is,
any
shortest tour of~${\vertices(P)}$.
A~\emph{scaling} of a point set~$\mathcal{P}$
\emph{from} a point~$c$ \emph{with scaling factor~$\alpha$} is the set of
points~${\{c+\alpha(p-c)\mid p\in\mathcal{P}\}}$. %
A \emph{fully-dimensional polytope} is a bounded and fully-dimensional
polyhedron, that is, a polyhedron that is bounded and contains a~$d$-dimensional
ball of strictly positive radius.
Unless otherwise specified, we use hyperplane, 
hypercube, etc.\ to refer to the corresponding objects in the~$d$-dimensional space.

\paragraph{Preliminary Observations.} 

Let~$T$ be an optimal tour for~$\mathcal{I}$ and suppose that we know the polytope~${\conv(T)}$. By the following lemma it suffices to find an optimal tour for the vertices of~${\conv(T)}$. 
\begin{lemma}\label{lem:prelim1}
 Let~$P$ be any convex polytope.
 Every tour of~$P$ is a feasible solution to~$\mathcal{I}$ if and
 only if~$P$ intersects every hyperplane of~$\mathcal{I}$. 
\end{lemma}
\begin{proof}
	Let~$P$ be a convex polytope that intersects every hyperplane of~$\mathcal{I}$,
        and~$T$ be any tour on the vertices of~$P$.
	Consider any hyperplane~$h$ of~$\mathcal{I}$. If it contains any vertex of~$P$, then it is visited by~$T$. 
	If it does not contain any vertex of~$P$, then it must intersect the interior of~$P$ and separate at least one pair of vertices of~$P$.
	Hence, any path connecting that pair must intersect~$h$, thus,
        the tour~$T$ visits~$h$ in this case as well.

        For the only-if part, assume that there exists a hyperplane~$h$ such that~$P$ does not intersect~$h$. Since by convexity any tour on
        the vertices of~$P$ is contained in~$P$, no such tour can
        visit~$h$.
 \end{proof}

\begin{corollary}\label{cor:prelim1}
Any optimal tour 
of~${\conv(\opt(\mathcal{I}))}$ is also
an optimal tour of~$\mathcal{I}$. 
\end{corollary}

Since we identify tours with polytopes and vice versa, we say that a polytope is \emph{feasible} if it intersects 
all input hyperplanes, and we say that a polytope is \emph{optimal} if it is the convex hull of an optimal tour.

\section{Structural Results}
\label{sec:struct}

The goal of this section is to prove that it suffices to focus our
attention only on solutions that visit the vertices of the polytopes
in~${\para(\grid_g)}$, where the specific size~${g=O_{\eps,d}(1)}$ will be
specified later.
Our proof boils down to showing that an arbitrary optimal tour can be
transformed into one that visits only the vertices of such a polytope~${P\in\para(\grid_g)}$, while not
increasing the length of the tour by more than a~${(1+\epsilon)}$-factor. This bound then will imply that the algorithm
presented in Section~\ref{sec:algo} is indeed a PTAS.
More specifically, the goal of this section is to show the following lemma.

\begin{lemma}\label{lem:finite-grid-enough}
  For any fixed~${\eps>0}$, there exists a grid~$\grid_g$ of size~${g=O_{\epsilon,d}(1)}$ such that
  for any input set~$\mathcal{I}$ of hyperplanes
  there is a polytope~${P\in\para(\grid_g)}$ with~${T=\tsp(P)}$ 
  being feasible and 
\begin{align*}
  |T|\leq(1+\eps)\cdot\opt(\mathcal{I})\,. 
\end{align*}
\end{lemma}

In other words, Lemma~\ref{lem:finite-grid-enough} shows that
in order to obtain a~${(1+\eps)}$-approximate solution it suffices to
find a minimum-tour-length feasible polytope, among the  polytopes
in~${\para(\grid_g)}$, for an apropriate~$g$.

\subsection{Transformation of the Optimal Polytope to a Polytope of
  Bounded Complexity}
\label{sec:parallel-polytope}

In this subsection, we prove Lemma~\ref{lem:finite-grid-enough} using
the following theorem, which may be of independent interest. 
We prove the theorem in the next subsection.

\begin{theorem}\label{thm:finite-vertices-enough}
  Let~${\eps>0}$. There is a number~${k_{\eps,d}\in O_{\eps,d}(1)}$ such
  that, for any 
  convex polytope~$P$ in~$\mathbb{R}^d$, 
  there exists a (center) point~$c$ in~$P$ and a set~$\mathcal{P}$ of at
  most~${k_{\eps,d}}$ vertices of~$P$ with the following property: 
  The polytope~$P$ is a subset of~${\conv(\mathcal{P'})}$ 
  where~${\mathcal{P'}}$ is obtained by scaling~$\mathcal{P}$ from the center~$c$ with the scaling factor~${1+\eps}$.
\end{theorem}
Theorem~\ref{thm:finite-vertices-enough} 
implies that, by increasing
the size of the convex polytope by at most a~${(1+\eps)}$ factor, we can
focus our attention on a convex polytope with few vertices.
Working with a polytope with few vertices 
is crucial, since we increase the length of the 
tour by a small amount around each vertex. In order to keep this 
from accumulating to a too large increase, we need the number of vertices to 
be~${O_{\eps,d}(1)}$ and independent of the grid size so that we can
later pick the granularity of the grid small enough.

We also use the following observation
to prove Lemma~\ref{lem:finite-grid-enough}.

\begin{observation}\label{obs:tours} 
 Let~$P$ be any convex polytope and let~$\mathcal{P}$ be any subset of its vertices. 
   If we scale~$\mathcal{P}$ from any center~$c$ by a scaling factor~${\alpha>1}$, then the following holds for the resulting point set~${\mathcal{P'}}$: 
	\[
		|\tsp(\conv(\mathcal{P'}))| \le \alpha |\tsp(P)|\,.
	\]
\end{observation}

\begin{proof}
	Consider~${\tsp(P)}$ and shortcut it to a tour~$T$ that visits
	only the vertices in~$\mathcal{P}$. 
	Scaling~$T$ by the factor~$\alpha$ 
	from the center~$c$ gives us a
	tour~${T'}$ that visits all vertices in~${\mathcal{P'}}$. By the intercept theorem, we
	have~${|T'|=\alpha|T|\le \alpha|\tsp(P)|}$. 
\end{proof}

The main idea in the proof of Lemma~\ref{lem:finite-grid-enough} is based on
transforming the polytope~$P$ of an optimal tour~$\opt$ into a polytope~${P'}$
in~${\para(\grid_g)}$.  First, we sparsify~$P$ and scale it up by the
factor~${1+\eps}$, as described in Theorem~\ref{thm:finite-vertices-enough}
and then we 
\enquote{snap} it to a~$d$-dimensional grid~${\grid_g'}$ to get~${P'}$,
where~${\grid_g'}$ can be obtained from~$\grid_g$ via translating
and scaling. This will directly imply that~${P'\in\para(\grid_g)}$.
We finish the proof by additionally showing
that~$P$ is contained in~${P'}$ (which implies the feasibility of~${P'}$), and that the length of~${\tsp(P')}$ is not much bigger
than the length of~$\opt$.

\begin{proof}[Proof of Lemma~\ref{lem:finite-grid-enough}]
  	Fix an~${\eps'>0}$ such that it fulfills~${1+\eps\ge (1+\eps')^2}$. 
  	Let
  	\begin{align*}
  	g = \left\lceil\frac{k_{\eps',d}\cdot\sqrt{d}(2^d+1)}{\eps'}\right\rceil
  	\end{align*} be the size of the grid~$\grid_g$  	(note that~${g=O_{\eps,d}(1)}$). 
  	Fix an arbitrary
	input instance~$\mathcal{I}$ and 
	let~$\opt$ be an optimal tour
	for~$\mathcal{I}$ with length~${|\opt|}$. 
	Let~${P = \conv(\opt)}$; we call it the \emph{optimal polytope}. By
	Corollary~\ref{cor:prelim1},~${\tsp(P)}$ is an optimal solution
	for~$\mathcal{I}$. 
	We apply Theorem~\ref{thm:finite-vertices-enough} on~${\eps'}$ and~$P$ and
	obtain the scaled point set~${\mathcal{P'}}$.
	Let~${P'= \conv(\mathcal{P}')}$. 
	By Theorem~\ref{thm:finite-vertices-enough},~${P\subseteq P'}$ and, hence,~${P'}$
	intersects every hyperplane in~$\mathcal{I}$. Thus, by
	Lemma~\ref{lem:prelim1},~${\tsp(P')}$ is a feasible tour of~$\mathcal{I}$ and, by
	Observation~\ref{obs:tours}, the tour is not too expensive. However, this does
	not prove Lemma~\ref{lem:finite-grid-enough} yet, as~${P'}$ is not necessarily
	contained in~${\para(\grid_g)}$. 
	In order to achieve this, we start by defining a
        grid~${\grid_g'}$ such that it can be obtained by scaling and
        translating~$\grid_g$. 
	It will help us to transform~${P'}$ to another polytope~${P''}$ that has some desirable
	properties.

	Consider the smallest possible axis-aligned bounding hypercube of~${P'}$. Let~${D'}$
	be its edge length. 
The grid~${\grid_g'}$ is now obtained by applying to the hypercube an
axis-aligned~$d$-dimensional grid of granularity (grid-cell side length)
	\begin{align*}
	    D'/g \le  D' \cdot \frac{\eps'}{k_{\eps',d}\cdot \sqrt{d}(2^d+1)}. 
	\end{align*}
	Again, note that~${\grid_g'}$ can be obtained by translating and
        scaling 
	 the grid~$\grid_g$ (by the factor~${D'/g}$.
         Hence, by this fact and the definition of~${\para(\grid_g)}$, 
	any fully-dimensional polytope that has its vertices at grid-points
	of~${\grid_g'}$ is also in~${\para(\grid_g)}$.
 
Thus,
	it suffices to transform the polytope~${P'}$ into a polytope~${P''}$ such that
  	\begin{enumerate}[label=(\roman*)]
	  	\item $P''$ has its vertices at grid-points of~${\grid_g'}$, \label{item:grid-points}
	  	\item $P''\supseteq   P'$, and \label{item:supset}
	  	\item $|\tsp(P'')|\le (1+\eps')|\tsp(P')|$. \label{item:cost} 
 	\end{enumerate}

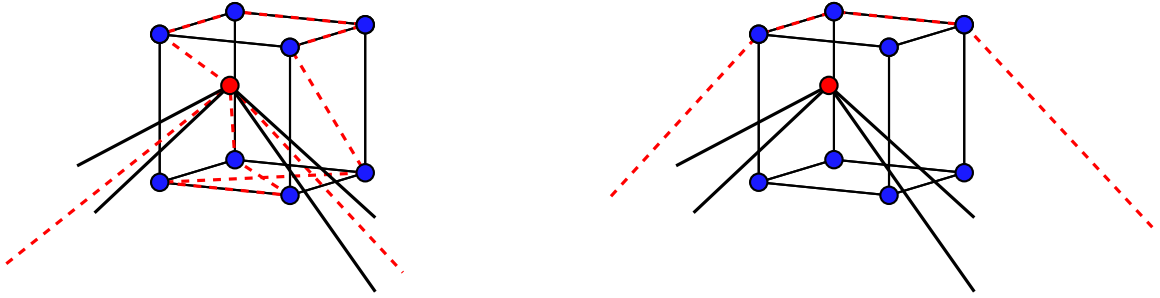
\begin{figure*}
	\centering
	\begin{subfigure}[t]{0.46\textwidth}
		\centering
		\begin{tikzpicture}[scale=1,
		tdplot_main_coords,axis/.style={->},thick]  
		
		\draw[color=red,dashed, very thick](1,0.5,1.2) -- (-1, 2, -1.5);
		
		\foreach \x in {0,2}
		\foreach \y in {0,2}
		\foreach \z in {0,2}
		{
			\draw[thick,opacity=1] (\x,0,\z) -- (\x,2,\z);
			\draw[thick,opacity=1] (0,\y,\z) -- (2,\y,\z);
			\draw[thick,opacity=1] (\x,\y,0) -- (\x,\y,2);
		}
		
		\draw[color=red, dashed, very thick] (4, -1.2, -0.9) -- (1,0.5,1.2) -- (2,0,2)
		-- (0, 0, 2) -- (0, 2, 2) -- (2, 2, 2) -- (0,2,0) -- (2,0,0) -- (2,2,0) --
		(0,0,0)  -- (1,0.5,1.2);

		\foreach \x in {0,2}
		\foreach \y in {0,2}
		\foreach \z in {2}
		{\draw[fill=blue!90] (\x,\y,\z) circle (0.3em);}    
		
		\foreach \x in {0,2}
		\foreach \y in {0,2}
		\foreach \z in {0,2}
		{\draw[fill=blue!90] (\x,\y,\z) circle (0.3em);}

		\draw[very thick] (-2,1,-1) -- (1,0.5, 1.2) -- (-1, -3, -0.5);
		\draw[very thick] (2,-1,-0.5) -- (1,0.5, 1.2) -- (-2, 1, -2);
		\draw[fill=red] (1, 0.5, 1.2) circle (0.3em);
		\end{tikzpicture}
		\caption{The tour gets modified to have a loop for each vertex
			going through all vertices of the
                        hypercube.
                      }
	\end{subfigure}
	\hfill
	\begin{subfigure}[t]{0.45\textwidth}
		\centering
		\begin{tikzpicture}[scale=1,
		tdplot_main_coords,axis/.style={->},thick] 
		
		\foreach \x in {-10,-8}
		\foreach \y in {0,2}
		\foreach \z in {0,2}
		{
			\draw[thick,opacity=1] (\x,0,\z) -- (\x,2,\z);
			\draw[thick,opacity=1] (-10,\y,\z) -- (-8,\y,\z);
			\draw[thick,opacity=1] (\x,\y,0) -- (\x,\y,2);
		}
		
		\draw[color=red, dashed, very thick]  (-11, -4, -1) --  (-8,0,2)
		-- (-10, 0, 2) -- (-10, 2, 2) -- (-15,2,-1.5);

		\foreach \x in {-10,-8}
		\foreach \y in {0,2}
		\foreach \z in {2}
		{\draw[fill=blue!90] (\x,\y,\z) circle (0.3em);}    
		
		\foreach \x in {-8,-10}
		\foreach \y in {0,2}
		\foreach \z in {0,2}
		{\draw[fill=blue!90] (\x,\y,\z) circle (0.3em);}

		\draw[very thick] (-12,1,-1) -- (-9,0.5, 1.2) -- (-11, -3, -0.5);
		\draw[very thick] (-8,-1,-0.5) -- (-9,0.5, 1.2) -- (-12, 1, -2);
		\draw[fill=red] (-9, 0.5, 1.2) circle (0.3em);        
		\end{tikzpicture}
		\caption{The new tour only goes along vertices in the convex
			hull.}
	\end{subfigure}
	\caption{Snapping a polytope to the grid. A 3D-example. The tour is
		the dashed line. After taking the convex hull of the new points, the
		tour can be shortcut: see cube on the right. 
		\label{fig:snapping}
	}
\end{figure*}

The transformation consists of mapping each vertex~${v\in\vertices(P')}$
to a subset of~$2^d$ many vertices of the grid~${\grid_g'}$. More
specifically, we map~$v$ to~${\vertices(\mathcal{C}_v)}$ where~${\mathcal{C}_v}$ is a closed hypercube cell
of~${\grid_g'}$ that contains~$v$. 
 The polytope~${P''}$ is then simply defined
as the convex hull of the grid points to which we mapped~${\vertices(P')}$.

	In order to conclude the proof, we show the three desired properties.
	Property~\ref{item:grid-points} directly follows by construction: Each point in~${\mathcal{P}''}$ is a vertex of the grid~${\grid_g'}$.
	Property~\ref{item:supset} directly follows from the facts that each vertex~$v$ of~${\mathcal{P}'}$
	is within a hypercube cell~${\mathcal{C}_v}$ of the grid~${\grid_g'}$, and that
	each vertex~${v'\in\mathcal{C}_v}$ is also in~${\mathcal{P}''}$.
	In turn,~$v$ is contained in~${\conv(\mathcal{P}'') = P''}$. 
	In order to prove Property~\ref{item:cost}, consider the following tour 
	on~${\vertices(P')\cup\vertices(P'')}$: We start with the tour~${\tsp(P')}$,
	and for each vertex~${v\in\vertices(P')}$ that it visits, we insert a 
	tour starting at~$v$ and visiting consecutively all vertices of~${\vertices(\mathcal{C}_v)}$ before returning to~$v$ and
	continuing with~${\tsp(P')}$. See also
	Figure~\ref{fig:snapping}. 
	This clearly increases the length of~${\tsp(P')}$ by at most~${({D'}/{g})\sqrt{d}(2^d+1)}$ for each visited vertex~${v\in\vertices(P')}$.
	Since~${|\vertices(P')| = k_{\eps',d}}$ by Theorem~\ref{thm:finite-vertices-enough}, 
	the total increase in cost is at 
	most~${({D'}/{g}) k_{\eps',d}\sqrt{d}(2^d+1) \le \eps' D'}$. 
	We can now
	remove~${\vertices(P')}$ from the resulting tour by shortcutting which
	can only decrease the total length. This way we obtain a tour
        for~${P''}$. 
	Therefore, 
	\begin{align*}
		|\tsp(P'')| & \le |\tsp(P')|+\eps' D'\\
		&\le (1+\eps')|\tsp(P')| \\
		&\le  (1+\eps')^2|\tsp(\opt)|\\
		&\le (1+\eps)|\tsp(\opt)|\,.
	\end{align*} 
	The second inequality
	follows because, by the fact that the smallest bounding box has
	edge length~${D'}$, we have~${D'\le |\tsp(P')|}$. The third inequality
	follows by Observation~\ref{obs:tours} when setting~${\alpha=1+\eps'}$.
\end{proof}

\subsection{Reducing the Number of Vertices}
\label{sec:reducing-to-finite-number}

In this subsection, we prove Theorem~\ref{thm:finite-vertices-enough}. 
Let~$P$ be a convex polytope. Among the vertices of~$P$, we  
identify a subset~$V$ of size~${k_{\eps,d}}$ and 
show that~$V$ fulfills the desired properties. 

In order to obtain some intuition, assume that the maximum-volume hyperellipsoid
contained in~$P$ 
is a hypersphere, and
that the minimum-volume hypersphere that has the same center and
contains~$P$ 
is not much bigger than that
hypersphere (as Lemma~\ref{lem:ball} shows, this is without loss of generality). 
Then we know that~$P$ has a ``regular'' shape not too far from the
internal and the external hypersphere. We can use this, along with
results from the core-set framework for extent
measures~\cite{coreset-survey,Chan06}, to
identify a subset~$V$ of the vertices, of constant cardinality, 
that defines a new polytope. 
From this new polytope we construct, by 
scaling 
from a center, 
a polytope of constantly many vertices, that is ``not too far'' from~$P$
(see Lemma~\ref{lem:sparcify}). The proof of
Theorem~\ref{thm:finite-vertices-enough} directly follows from
combining the facts above.

We use an auxiliary technical lemma, which is an extension of the
following  result of Ball~\cite{Ball1992}. For the remainder of the
section, let~${B(c,r)}$ denote
the hypersphere with center~$c$ and radius~$r$ either in the entire~$d$-dimensional space or within an affine subspace implied by the context.
 
\begin{lemma}[Ball~\cite{Ball1992}, Remarks]
  \label{lem:balls-own}
  If~$C$ is a convex body whose contained~$d$-dimensional hyperellipsoid of maximal volume
  is~${B(c,1)}$ (for some center~$c$), then the diameter of~$C$ is upper
  bounded by~${\sqrt{2d(d+1)}}$.
\end{lemma}

The following auxiliary lemma relates a polytope~$P$ to specific
hyperellipsoids and hyperspheres in~${\mathbb{R}^d}$.

\begin{lemma}
  \label{lem:ball}
  For any convex polytope~$P$ with~${d'}$-dimensional volume, where~${d'\le d}$ is maximal, let~$S$ be the~${d'}$-dimensional affine subspace of~${\mathbb{R}^d}$ containing~$P$. 
  There 
  exists a center~${c\in P}$ and 
  an affine transformation~${\mathcal{F}:S\rightarrow S}$ such
  that~${\mathcal{F}(P)}$ 
  has the following properties:
  \begin{itemize}
  \item The maximum-volume~${d'}$-dimensional hyperellipsoid contained in~${\mathcal{F}(P)}$ 
    is~${B(c,1)}$.
  \item $B(c,d')$ contains~${\mathcal{F}(P)}$.
  \end{itemize}
\end{lemma}

\begin{proof}
  Consider the maximal-volume~${d'}$-dimensional hyperellipsoid~$\mathcal{E}$ contained in~$P$, and let its center be~$c$. 
  Let~$\mathcal{F}$ be the affine transformation\footnote{A hyperellipsoid is defined as the image 
  	of a hypersphere under invertible linear transformations.}
  that transforms~$\mathcal{E}$ into~${B(c,1)}$. 

  Since~$P$  
  contains~$\mathcal{E}$, 
 the transformed polytope~${\mathcal{F}(P)}$ must contain~${B(c,1)}$. 
  Let~${P'= \mathcal{F}(P)}$.
  In order
  to  show the first property, we still need to prove that~${B(c,1)}$ is
  the maximal-volume hyperellipsoid contained in~${P'}$. 
  So suppose 
  that
  there exists a hyperellipsoid~${\mathcal{E}'}$ contained in~${P'}$ with volume strictly larger than~${B(c,1)}$. 
  But then we could apply
  the affine transformation~${\mathcal{F}^{-1}}$ to~${\mathcal{E}'}$ and
  obtain a hyperellipsoid contained in~$P$ of larger volume 
  than~$\mathcal{E}$ (the fact that the volume order is preserved
  follows by a simple change of variables and because the
  transformation is non-degenerate), a
  contradiction. 

The second property directly follows by Lemma~\ref{lem:balls-own} by
setting~${d = d'}$ and observing that~${d'\ge 1}$.
\end{proof}

It remains to show how one can ``sparsify'' a given
polytope, that is, how one can reduce the number of vertices of the
polytope to~${O_{\eps,d}(1)}$ many, while still maintaining a set of
desirable properties. In order to do this, we employ a result from the
framework of core-sets~\cite{coreset-survey,Chan06}. In order to keep
the paper self-contained and to simplify cross-checking, we include
the related definitions and results explicitly, although we only
need parts of them for our results.

\begin{Definition}[$\epsilon$-core-set~\cite{Chan06}]
  Given a double-argument measure~${\mu(\cdot,\cdot)}$ 
we say that a subset~${R\subseteq V}$ is an~$\epsilon$-core-set of~$V$ (over a set~$Q$) 
if~$R$ is of constant size and~${\mu(V,x)\ge \mu(R,x)\ge (1-\epsilon)\mu(V,x)}$ for all~$x$ (in~$Q$). 
\end{Definition}

We are only interested in one particular measure, the extent measure:

\begin{Definition}[Extent measure~\cite{Chan06}]
For some point set~$V$ in the~$d$-dimensional Euclidean space, let
the \emph{extent measure} with the respect to a direction vector~$\vec{x}$ be~${w(V,\vec{x}):= \max_{p,q\in V}(\vec{p}-\vec{q})\cdot\vec{x}}$. 
The \emph{one-sided extent measure} is similarly defined as~${\overline{w}(V,\vec{x}) := \max_{p\in V} \vec{p}\cdot\vec{x}}$.
\end{Definition}

The following results are known for~$\epsilon$-core-sets for the extent measure:

\begin{theorem}[\cite{Chan06}]\label{thm:finding-coresets}
  Given an~$m$-point set in~${\mathcal{R}^d}$, one can construct an~$\epsilon$-core-set of size~${O(1/\epsilon^{(d-1)/2})}$ for the
  extent measure in time~${O(m+1/\epsilon^{d-3/2})}$ 
  or time~${O((m+1/\epsilon^{d-2})\log (1/\epsilon))}$.
\end{theorem}

\begin{observation}[\cite{Chan06}]\label{obs:coresets-inqlty}
  If~$R$ is an~$\epsilon$-core-set for~$V$ for the extent measure,
  then the following holds:
\begin{align*}
  \overline{w}(R,\vec{x})\ge \overline{w}(V,\vec{x})-\epsilon w(V,\vec{x}).
\end{align*}
\end{observation}

We are now ready to prove the following lemma which is helpful in
``sparsifying'' a given polytope.

\begin{lemma}\label{lem:sparcify}
  Let~${\eps>0}$, and let~$P$ be a 
  convex polytope with~${d'}$-dimensional volume, where~${d'\le d}$ is maximal, 
  spanning a~${d'}$-dimensional affine subspace~$S$ of~${\mathbb{R}^d}$ such that: 
  \begin{align*}
    B(\vec{0},1) \subseteq P \subseteq B(\vec{0},d').
  \end{align*}
  There is a set~${V'}$ of points with the following properties:
  \begin{enumerate}[label=(\roman*),font=\normalfont]
  \item For each point~${v'\in V'}$, there exists a vertex~$v$ of~$P$ with~${\vec{v}'=(1+\eps)\vec{v}}$, \label{item:vertices}
  \item $|V'|=k_{\eps,d} \in O_{\eps,d}(1)$, \label{item:onlyfew}
  \item $P \subseteq \conv(V')$. \label{item:contained} 
  \end{enumerate}
\end{lemma}

\begin{proof}

Let~${\epsilon':= \epsilon/((1 + \epsilon)2d')}$ and 
let~$R$ be an~${\epsilon'}$-core-set obtained by applying
Theorem~\ref{thm:finding-coresets} to the set~$V$ of vertices of~$P$. 
Scale each point in~$R$ by the factor~${1+\epsilon}$ from the
origin, that is, for all~${v\in R}$ set~${\vec{v'} = (1+\epsilon)\vec{v}}$ and obtain the set~${V'}$. 
Note that, by construction,~${V'}$
satisfies Property~\ref{item:grid-points}, and furthermore also Property~\ref{item:onlyfew} since~${|V'|=|R|= O(1/{\epsilon'}^{(d'-1)/2})=O_{\eps,d}(1)}$, and thus it remains to show Property~\ref{item:contained}.

Suppose for the sake of contradiction that there exists a point~${p'\in P}$ 
such that~${p'\not\in \conv(V')}$. Given the convexity of both~${\conv(V')}$ and~${p'}$, we know by the separation theorem that there
must exist a hyperplane~$h$ separating~${\conv(V')}$
and~${p'}$. Let~$\vec{x}$ be a normal vector of~$h$ oriented towards
the side of~$h$ that contains~$p'$,
 and let~${p=\arg\max_{p\in P}\vec{p}\cdot\vec{x}}$. 
Thus,~${\vec{p}\cdot\vec{x}\ge \vec{p'}\cdot \vec{x}}$ and consequently~$p$ lies on the same side of~$h$ as~${p'}$, which means that~$p$ is also separated from~${\conv(V')}$ by~$h$.

Hence, by Observation~\ref{obs:coresets-inqlty} we have
\begin{align*}
\overline{w}(R,\vec{x}) &\ge \overline{w}(V,\vec{x}) -
\epsilon'w(V,\vec{x})\\
&\ge \vec{p}\cdot\vec{x}  -
\epsilon'2d'\\
&\ge \vec{p}\cdot\vec{x} - \frac{\epsilon}{1+\epsilon},
\end{align*}
where the last two inequalities hold by the definitions of~$p$ and~${\epsilon'}$ and by the containment~${\conv(V)\subseteq B(\vec{0},d')}$.
We therefore
have
\begin{align*}
  \overline{w}(V',\vec{x}) &= (1+\epsilon)\overline{w}(R,\vec{x})\ge
  (1+\epsilon)\left( \vec{p}\cdot\vec{x} -
    \frac{\epsilon}{1+\epsilon}\right)\\ &=
  (1+\epsilon)\vec{p}\cdot\vec{x} - \epsilon = \vec{p}\cdot\vec{x}
  +\epsilon(\vec{p}\cdot\vec{x} - 1) \ge \vec{p}\cdot\vec{x},
\end{align*}
where the first equality holds by the definition of~${V'}$ and the
linearity of the dot product of Euclidean vectors, the first
inequality by substituting~${\overline{w}(R,\vec{x})}$ from above,
and the last inequality because~${\vec{p}\cdot\vec{x}\ge 1}$ since~${P\supseteq B(\vec{0},1)}$ and~${p\in P}$. 
However, this is a contradiction, since by the fact that~$p$
and~$\conv(V')$ lie on different sides of~$h$ and given that~$\vec{x}$
is oriented towards the side containing~$p$,
it follows that the one-sided extent measure of~${V'}$ with the respect to~$\vec{x}$ must be strictly less than~${\vec{p}\cdot\vec{x}}$, 
that is,~${\overline{w}(V',\vec{x}) < \vec{p}\cdot\vec{x}}$. The proof is illustrated in Figure~\ref{fig:sparcify}.
\end{proof}

\begin{figure*}
\centering
\begin{tikzpicture}[scale=0.52]
    \tikzstyle{every node}=[inner sep=0pt, square/.style={regular polygon,regular polygon sides=4}]
		\coordinate (O) at (0,0);
		
		\coordinate (B) at (-4.5,8.5);
		\coordinate (Q) at ($ (O)!1.60!(B) $);
		\coordinate (R) at (-6,7);
		\coordinate (G) at ($(B)!2.8!(R)$);
		
		\coordinate (F) at (7.5,8);
		\coordinate (P) at ($ (O)!1.60!(F) $);
		\coordinate (S) at (12,5.5);
		\coordinate (H) at ($(F)!1.6!(S)$);
		
		\coordinate (E) at (4.5,11.5);
		\coordinate (M) at ($ (O)!1.30!(E) $);
		\coordinate (T) at (9,10.14);
		\coordinate (J) at ($(E)!1.7!(T)$);
		
		\coordinate (Y) at (2.5,10);
		\coordinate (N) at ($ (O)!1.30!(Y) $);
		
		\coordinate (D) at (0,9.3);
		\coordinate (Z) at (-1.8,3);
		\coordinate (W) at ($(D)!1.45!(Z)$);

      \draw[dashed] (G) -- (R);
      \draw (R) -- (B);
      
      \draw[dashed] (H) -- (S);
      \draw (S) -- (F);
      
      \draw[dashed] (J) -- (T);
      \draw (T) -- (E);
      
      \draw[dashed] (W) -- (Z);
      \draw (Z) -- (D);
      \draw (B) -- (D) -- (E) -- (F) -- (D);
      \draw (E) -- (B);
      \draw[dashed] (O) -- (B);
      \draw[dashed] (O) -- (F);
      \draw[dashed] (O) -- (E);
      \draw (B) -- (Q);
      \draw (F) -- (P);
      \draw (E) -- (M);
      \node () at (0.4,-0.3) {$\vec{0}$};
      \draw[fill=blue] (B) circle (0.5em);
      \draw[fill=blue] (F) circle (0.5em);
      \draw[fill=blue] (E) circle (0.5em);
      \node (EB) at ($(B)!.3!(Q)$) [square, fill=red, inner sep=0.15em, draw] {};
      \node (EF) at ($(F)!.25!(P)$) [square, fill=red, inner sep=0.15em, draw] {};
      \node (EE) at ($(E)!.4!(M)$) [square, fill=red, inner sep=0.15em,
      draw] {};
      \draw[very thick,dashed,green] (O) -- (Y);
      \draw[fill=gray] (Y) circle (0.5em);
      \node at (2.9,9.7) {$p$};
      \draw[very thick,dashed,green,->] (Y) -- ($(O)!1.5!(Y)$);
\end{tikzpicture}
\caption{An example in 3D: The blue (round) vertices correspond to points in
  the core-set~$R$, and the red (square) vertices are obtained by scaling the blue ones 
  by the factor~${1+\epsilon}$ from the origin. The thick dashed line
  through~$p$ is the direction of~$\vec{x}$.}
\label{fig:sparcify}
\end{figure*}
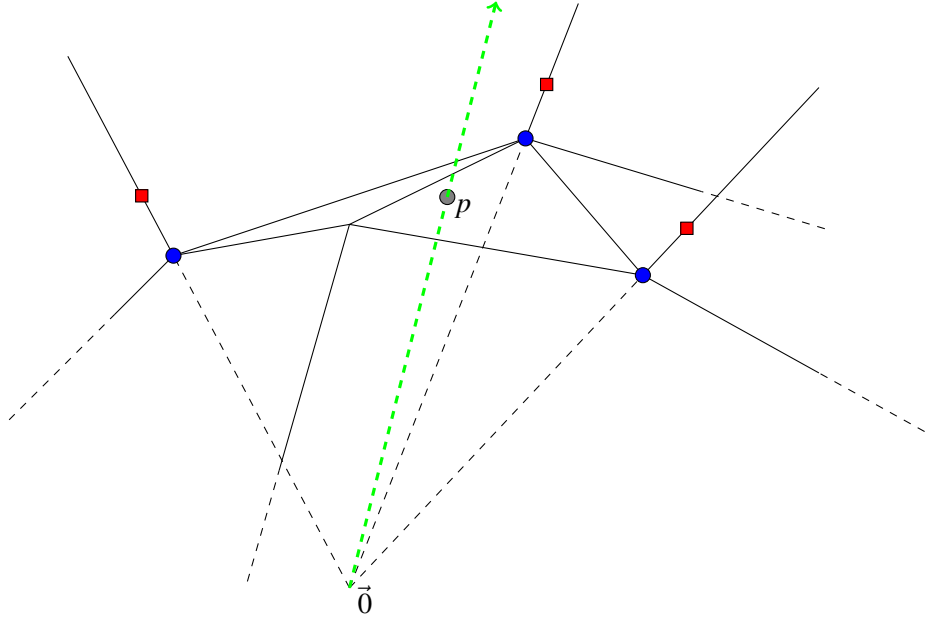

We are now ready to prove Theorem~\ref{thm:finite-vertices-enough}.

\begin{proof}[Proof of Theorem~\ref{thm:finite-vertices-enough}]
   Let~$P$ be any convex polytope in~$\mathbb{R}^d$. Apply Lemma~\ref{lem:ball}
to~$P$. This gives us~${P'}$ and a center~${c}$. Now, translate~${P'}$ by~${-\vec{c}}$ such that the
center of the two hyperspheres is~$\vec{0}$. Let~${P''}$ be this translate 
of~${P'}$. We apply Lemma~\ref{lem:sparcify} to~${P''}$ which in turn gives us a
set of points~${V'}$, that satisfy Properties~\ref{item:vertices},~\ref{item:onlyfew}, and~\ref{item:contained} from
Lemma~\ref{lem:sparcify}. The vertices of~$P$ corresponding to~${V'}$ satisfy
the theorem. 
\end{proof}

\section{Our Algorithm}
\label{sec:algo}
In this section, we show that there exists a strongly polynomial linear
time algorithm that finds a tour of minimum
length among all the tours that visit the vertices of some feasible polytope
in~${\para(\grid)}$. More formally, we show the following lemma.

\begin{lemma}\label{lem:alg} 
	Let~$\grid$ be a given integer grid of constant size. There is a strongly
polynomial linear time algorithm that computes for any set~$\mathcal{I}$ of
input hyperplanes a feasible tour~$T$ of minimum length such
that~${\conv(T)\in\para(\grid)}$. 
\end{lemma}

Note that together with Lemma~\ref{lem:finite-grid-enough}, Lemma~\ref{lem:alg} yields Theorem~\ref{thm:main}.

\subsection{Overview of the algorithm}
\label{subsec:alg-overview}

The idea of the algorithm is to construct several linear programs, such that the
one that finds the shortest tour among all of them provides the actual shortest
tour of all possible tours that are tours of the vertices of some polytope
in~${\para(\grid)}$. We construct one LP for each of the constantly many base
polytopes of~$\grid$, which we simply enumerate. Each of the LPs has~${d+1}$ variables and a number of
constraints that is linear in~$n$u. Note that, given a base
polytope, we can compute the optimal tour of its vertices, that is, the order in
which its vertices are visited. Moreover, this order is the same for any scaled
and translated polytope, and the length of the tour is scaled equally to the
polytope.

We next describe the construction of the LPs in more detail. For a given base
polytope~$\gridpoints$, the LP maintains a scaling variable~$\scalevar$ and a
vector~$\translvar$ of translation variables. These variables represent a linear
transformation of~$\gridpoints$ in which~$G$ is scaled by the
factor~$\scalevar$ (from the origin) 
and translated by the vector~$\translvar$. 
We assume that~$\gridpoints$ is given as a set of grid points of~$\grid$.

By Lemma~\ref{lem:prelim1}, the optimal tour of the scaled and translated
polytope is feasible if the polytope intersects each input hyperplane. To ensure
this, we use an idea similar to that of~\citet{TSPNlinesBallsPlanes}: For each
input hyperplane, we select two vertices of~$\gridpoints$ (the \emph{separated pair}) and write a
\emph{feasibility constraint} requiring the two vertices to be on different
sides of the hyperplane. These constraints ensure that the convex hull of the vertices
intersects each hyperplane and thus any tour that visits all its vertices is
feasible (Lemma~\ref{lem:prelim1}). For each input hyperplane~${i\in\mathcal{I}}$,
we let the separated pair consist of two vertices~${p,q\in\gridpoints}$ for which
there are translates~$i_p$ and~$i_q$ of~$i$ such that~$p$ lies in~$i_p$ and~$q$
lies in~$i_q$ and~$G$ lies in the convex hull of these two hyperplanes.

In the rest of this section, we describe our construction of the LPs that
altogether take into consideration the whole search space consisting of all
polytopes in~${\para(\grid)}$. Finally, we proof Lemma~\ref{lem:alg} by proving
that the running time of the algorithm is strongly polynomial linear.

\subsection{The LPs}

Let~$\gridpoints$ be given as the set of its vertices. 
The LP scales~$\gridpoints$ by the variable factor~${\scalevar\in
  \R^+}$ from the origin 
and translates it by the
variable vector~${\vec{\translvar}\in\R^d}$.
We denote the scaled and shifted set
by~${\scalevar\gridpoints+\vec{\translvar}}$.

For each hyperplane~${i\in\mathcal{I}}$, given by a
normal vector~$\normvect_i$
and a value~$\val_i$, such that~${\val_i\normvect_i\in i}$,
we define the separated pair~${(\separate_{i}^{\;+},\separate_{i}^{\;-})\in \gridpoints^2}$ 
of~$i$ as 
\begin{align*}
\separate_{i}^{\;+}\in\argmax_{\vec{\gridpoint}\in\gridpoints}\dotprod{\vec{\gridpoint}}{\normvect_i}\qquad\text{and}\qquad
\separate_{i}^{\;-}\in\argmin_{\vec{\gridpoint}\in\gridpoints}\dotprod{\vec{\gridpoint}}{\normvect_i}\,.
\end{align*}

\begin{lemma} \label{lem:separated}
	The convex 
	hull 
	of~${\scalevar\gridpoints+\vec{\translvar}}$ intersects the hyperplane~${i\in\mathcal{I}}$ if and only if~${\dotprod{\scalevar\separate_{i}^{\;+}+\vec{\translvar}}{\normvect_i}\geq \val_i}$ and~${\dotprod{\scalevar\separate_{i}^{\;-}+\vec{\translvar}}{\normvect_i}\leq \val_i}$. 
\end{lemma}
\begin{proof}
	By the definitions of~${\separate_{i}^{\;+}}$
        and~${\separate_{i}^{\;-}}$, for any~${\scalevar\in \R^+}$
        and~${\vec{\translvar}\in\R^d}$, we have
	\begin{align*}
		\scalevar\separate_{i}^{\;+}+\vec{\translvar}&=\argmax_{\vec{\gridpoint}\in\scalevar\gridpoints+
		\vec{\translvar}}\dotprod{\vec{\gridpoint}}{\normvect_i}\qquad\text{and}\qquad 
		\scalevar\separate_{i}^{\;-}+\vec{\translvar}=\argmin_{\vec{\gridpoint}\in\scalevar\gridpoints+
		\vec{\translvar}}\dotprod{\vec{\gridpoint}}{\normvect_i}\,.
	\end{align*}
	Now, let~${P=\conv(\scalevar\gridpoints+\vec{\translvar})}$.
	Suppose that~$P$ does 
	intersect~$i$, then there is a point~${\gridpoint\in P}$ such that 
	\[
		\dotprod{\vec{\gridpoint}}{\normvect_i}=\val_i
	\]
	and therefore 
	\[
		\dotprod{\scalevar\separate_{i}^{\;+}+\vec{\translvar}}{\normvect_i}\ge \val_i
		\qquad\text{and}\qquad
		\dotprod{\scalevar\separate_{i}^{\;-}+\vec{\translvar}}{\normvect_i}\le \val_i\,.
	\]
	Now, suppose 
	\[
		\dotprod{\scalevar\separate_{i}^{\;+}+\vec{\translvar}}{\normvect_i}\geq \val_i\qquad\text{and}\qquad
		\dotprod{\scalevar\separate_{i}^{\;-}+\vec{\translvar}}{\normvect_i}\leq \val_i\,.
	\] 
	Then, 
	there is some convex combination~${\vec{\gridpoint}}$ of~${\scalevar\separate_{i}^{\;+}+\vec{\translvar}}$ and~${\scalevar\separate_{i}^{\;-}+\vec{\translvar}}$, such that 
	\[
		\dotprod{\vec{\gridpoint}}{\normvect_i}= \val_i
	\]
	and thus~$P$ intersects~$i$.
\end{proof}

From Lemma~\ref{lem:separated}, we obtain that the system of inequalities 
\begin{align*}
\dotprod{\scalevar\separate_{i}^{\;+}+\vec{\translvar}}{\normvect_i}&\ge \val_i&\forall i\in\mathcal I\\
\dotprod{\scalevar\separate_{i}^{\;-}+\vec{\translvar}}{\normvect_i}&\le \val_i&\forall i\in\mathcal I\\
\scalevar&\in \R^+\\
\vec{\translvar}&\in \R^d
\end{align*}
describes all polytopes~${P\in \para(\grid)}$ that correspond to a given base
polytope~$\gridpoints$ and intersect all hyperplanes in~$\mathcal{I}$. The
objective of the LP is simply to minimize the total length of the tour of the
computed set of vertices~${\scalevar\gridpoints+\vec{\translvar}}$. This is equivalent
to minimizing the scaling variable~$\scalevar$. The comparison of the outcomes
of two LPs for two different base polytopes~$\gridpoints_1$ and~$\gridpoints_2$
is then made by multiplying the objective values by~${|\tsp(\gridpoints_1)|}$
and~${|\tsp(\gridpoints_2)|}$, respectively.

\subsection{Proof of Lemma~\ref{lem:alg} and Discussion of the Running Time}
\label{subsec:proof-lemma}

\begin{proof}[Proof of Lemma~\ref{lem:alg}.]
	The number of LPs that we solve is equal to the number of 
	base polytopes of~$\grid$, which is constant. 
		 Since each LP has constant many variables and~${O(n)}$ constraints, where~${n=|\calI|}$, the running time of each LP is strongly polynomial linear in~$n$~\cite{DBLP:journals/jacm/Megiddo84, Chan16}.  
		 
		 It remains to show that we can select the shortest tour in linear time. 
		 In order to compare the Euclidean tour lengths, we compare sum of square roots (over integers by appropriate scaling).
		 In general, it is a long-standing open problem whether two such sums can be efficiently compared~\cite{rourke81openproblem,openProblemsProject33}. 
		 However, it is known that two such sums over~$k$ many square roots over integers from~${\{0,\dots,m\}}$ either differ by at least~${m^{-O(2^{k})}}$ or they are equal~\cite{Burnikel2000}. In our case,~$k$ is constant, thus 
		 it suffices to extract~${O_k(\log m)}$ digits of the square roots in order to compare two sums. 
		 By our assumption on the computational model, an extraction of this precision takes constant time\footnote{Computing a square root with precision~$n$ can be easily achieved by first multiplying the radicand by~$2^{2n}$, computing the integer square root and then dividing it by~$2^n$. Also note that the space is bounded by a polynomial in the input length and so is~$O_k(\log m)$.}, and thus comparing the lengths of two tours takes also constant time.
\end{proof}

If we consider the grid size~$g$ that we use in the proof of
Theorem~\ref{thm:main}, we obtain the following dependence of the
running time 
on the constants~$d$ and~$\eps$:
The number of base polytopes, that is, the number of LPs that we solve, can be generously bounded by~${2^{g^d}}$.
Given that the number of variables is~${d+1}$, each LP can be solved in time~${d^{O(d)} n}$~\cite{DBLP:journals/jacm/Megiddo84, Chan16}.
Thus, the running time of our algorithm is~$2^{g^d} d^{O(d)} n =
2^{O(1/\eps)^{d^2}}n$, where the equality follows from~${g=\left\lceil
    {\eps'}^{-1}k_{\eps',d}\cdot\sqrt{d}(2^d+1)\right\rceil =
  O(1/\eps')^{d} d^{O(1)} = O(1/\eps')^{d}}$
 and by choosing~${\eps'=O(\eps)}$.

\section{Extensions and Discussion}
\label{sec:discussion}
While we present a PTAS for TSP with hyperplane neighborhoods in this paper, the exact complexity status of the problem remains open. Even for $d$ as part of the input, it is not known whether the problem is NP-hard.

It would be interesting to find further applications of our techniques.
It is straightforward to extend our result to the TSPN path
problem with hyperplane neighborhoods where we want to visit all neighborhoods with a path instead of a tour. 
It is also not too hard to extend our result to cases where the tour or path is constrained to additionally visit some constantly many points or hyperplanes in a specific order (with the 
 respect to all hyperplanes).

For instance, consider the variant of the TSPN path problem where the input additionally specifies a start point~$s$: 
We guess for each base polytope the grid cell where the path should start, and then we assure via LP constraints that the shifted and scaled copy of this grid cell contains~$s$.
Then we connect~$T$ to~$s$.

Another version of the problem that is of interest and that
might admit a similar technique is the version where the input hyperplanes have to be
visited in a specific order. If furthermore the start point is given in the path version (as discussed above), then we arrive at 
the offline version of hyperplane chasing, which is an online problem where the hyperplanes are revealed one by one. 
Such online
chasing problems of convex bodies~\cite{FriedmanL93} and special cases as well as extensions thereof have received significant interest in the literature in recent years~\cite{chasing-bansal,chasing-argue,chasing-waoa,chasing-latin,chasing-approx,convex-bodies0,convex-bodies1,convex-bodies2}. 

Another promising direction for future research is trying to settle the complexity status of TSPN for other types of neighborhoods such as lower-dimensional affine subspaces and disks.

\paragraph{Acknowledgments.} The authors would like to thank Joseph
Mitchell for suggesting simplified versions of the proofs of
Lemma~\ref{lem:alg} and Theorem~\ref{thm:finite-vertices-enough}. We also thank several anonymous reviewers for their comments.

\bibliography{references}{}
\bibliographystyle{plainnat}

\end{document}